\documentclass[%
reprint,
nofootinbib,
amsmath,amssymb,
aps,
]{revtex4-1}
\pdfoutput=1
\usepackage[utf8]{inputenc}
\usepackage[english]{babel}
\usepackage[T1]{fontenc}
\usepackage{amsmath}
\usepackage{hyperref}

\usepackage{tikz}
\usepackage{lipsum}

\usepackage{dsfont}
\usepackage[normalem]{ulem}
\usepackage{mathtools}
\usepackage[makeroom]{cancel}
\usepackage{bbm}
\usepackage{enumitem}

\usepackage{color}
\usepackage{graphicx}
\usepackage{amsthm}
\definecolor{martin}{rgb}{0,.4,1}

\definecolor{henrik}{rgb}{1,.4,0}

\newcommand{\mb}[1]{\mathbb{#1}}

\newcommand{\id}{\mathbbm{1}}

\newcommand{\R}{\mb{R}}

\newcommand{\rom}[1]{\uppercase\expandafter{\romannumeral #1\relax}}

\usepackage{amssymb}             % AMS Math

\newtheorem{lemma}{Lemma}
\newtheorem{corollary}{Corollary}
\newtheorem{theorem}{Theorem}
\theoremstyle{definition}
\newtheorem{definition}{Definition}
\usepackage{algorithm}
\usepackage{algpseudocode}
\usepackage{comment}

\usepackage[normalem]{ulem}
\newcommand{\MF}{\mathcal{M}_{\mathcal{F}}}
\newcommand{\leqF}{\leq_{\mathcal{F}}}

\begin{document}
	
	\title{Undecidability in resource theory: can you tell theories apart?}
	\author{Matteo Scandi}\thanks{\href{mailto:matteo.scandi@icfo.eu}{Matteo.Scandi@icfo.eu}}
	\affiliation{ICFO - Institut de Ciencies Fotoniques, The Barcelona Institute of\\
		Science and Technology, Castelldefels (Barcelona), 08860, Spain}
	\author{Jacopo Surace}\thanks{\href{mailto:jacopo.surace@icfo.eu}{Jacopo.Surace@icfo.eu}}
	\affiliation{ICFO - Institut de Ciencies Fotoniques, The Barcelona Institute of\\
		Science and Technology, Castelldefels (Barcelona), 08860, Spain}
	\date{\today}
	
	%\orcid{0000-0003-0290-4698}
	%, \texttt{\textbackslash{}homepage}, and \texttt{\textbackslash{}thanks} commands to add additional information for the preceding \texttt{\textbackslash{}author}. If applicable, this can also be used to indicate that a work has previously been published in conference proceedings.}
	\begin{abstract}

		A central question in resource theory is whether one can construct a set of monotones that completely characterise the allowed transitions dictated by a set of free operations.  A similar question is whether two distinct sets of free operations generate the same class of transitions. 
		%These questions are part of the more general problem about the possibility of certifying the equivalence of two different characterisations of the same resource theory.
		These questions are part of the more general problem of whether it is possible to pass from one characterisation of a resource theory to another.
		 In the present letter we prove that in the context of quantum resource theories this class of problems is undecidable in general. This is done by proving the undecidability of the membership problem for CPTP maps, which subsumes all the other results.

		%In a resource theory not all the states are born the same. A restriction on the accessible set of quantum operations imposes a partial order of resourcefulness between states. A state is more resourceful then the other, if from the first it is possible to obtain the second via a permitted quantum operation. Following the same logic, one attributes to each state a set of numbers, called monotones, quantifying the amount of different resources the state has. The idea is that, operating with an accessible operation implies the consumption of some of the resources.  The allowed operations permits just to move from states with higher resources to states with lower resources.
		%It follows that the partial order relation between states, that is, the resource theory, has two different characterisation. In one case, a description of the allowed operations is provided; in the other, a quantification of the amount of resources of each state is specified.

		%In general one can expect that the task of assigning a set of free operations to some specifying resources must be hard if not at all impossible. The goal of these notes is to show that there's no algorithmic mean to assign to a resource theory a corresponding set of monotones that uniquely specifies it.
	\end{abstract}
	\maketitle
	%\onecolumngrid
	
	\section{Introduction}
	
	The aim of resource theory is to characterise the possibility of action of an agent who acts under some kind of operational restriction~\cite{chitambarQuantumResourceTheories2019}. To this end, one specifies a set of transformations that the agent can freely carry out, and asks general questions about its capabilities when it is assumed that the allowed operations can be composed repeatedly in any arbitrary order. 
	
	The paradigmatic example of a resource theory is the one of LOCC (local operations and classical communication)~\cite{horodeckiQuantumEntanglement2009}. In this case, two agents who can access two different halves of a shared quantum system are allowed to operate only through local transformations and by sharing classical information between them. It is a non-trivial fact that only by composing operations from the LOCC set, a quantum state can be perfectly teleported from one agent to the other~\cite{bennettTeleportingUnknownQuantum1993}. 
	
	The example above shows that being able to identify whether a transformation is part of a given resource theory is an issue of practical relevance, which can give quite surprising results. It should be noticed, though, that the membership of the quantum teleportation to the LOCC set  is proven by explicitly presenting a protocol to implement it. This kind of proofs require a certain amount of ingenuity and the \emph{ad hoc} constructions used do not help identifying generic members of a set. For this reason, the question of whether there is a general way to certify the membership of a transformation to a set of free operations remains open.
	
	A dual perspective about resource theories is given by focusing on states rather than operations. In this context, one assigns to each system a series of labels that quantify how useful the state is. The paradigmatic example is given by entanglement, the resource for LOCC operations. As a matter of fact, most of the non trivial protocols that can be carried out within LOCC (among these, the quantum teleportation protocol described above) are possible only by the use of entangled states. For this reason, it is also of practical importance to assess the value of a state within a resource theory. 
	
	The standard approach is to define a set of functions, called monotones, which do not increase under the application of free operations. In this way, one can estimate the resourcefulness of a state by looking at a family of numerical labels. 
	
	Hence, there exists two possible natural characterisations of a resource theory: in one case, a description of the allowed operations is provided; in the other, a set of value functions, the monotones, are specified. It is natural then to ask whether there is a constructive way to pass from one description to the other. This is equivalent to asking whether it is possible to pass from the description about the operational capabilities of an agent to the determination of the resourcefulness of a state.
	
	We prove here that both of the problems raised above are undecidable for general resource theories of quantum operations: namely we prove that there is no algorithm that decides whether a generic transformation is generated by a set of free operations, and we show that this implies that  there is no algorithmic means of constructing from a set of free operations a set of monotones describing the same resource theory. It also follows from the undecidability of the membership problem that it is impossible to certify whether a transition is part of a resource theory. Moreover, it is shown that given two sets of free operations it is impossible to tell whether they describe the same set of transitions. These negative results hint at the reason why finding a complete set of monotones is usually a difficult task for many particular resource theories.
	
	\section{Definitions}
	
	In this section we provide the main definitions of the objects treated in the rest of the paper. 
	
	Given a set $\mathcal{S}$ and an associative binary operation on it, the semigroup $\mathcal{S}^*$ is defined as the union of all the finite compositions of elements from $\mathcal{S}$. In other words, $\mathcal{S}$ is the generating set of $\mathcal{S}^*$.

	\begin{definition}[Resource theory]
		Given a set of operations $\mathcal{F}$ which contains the identity, the semigroup $\mathcal{F}^*$ characterises a resource theory. The elements of $\mathcal{F}^*$ are called free operations.
	\end{definition}
	
	%The resource theory $\mathcal{F}^*$ is completely characterised by $\mathcal{F}$, defined as the interesection of all the sets that generates $\mathcal{F}^*$. 
	
	In the following, we only consider the case in which the elements of $\mathcal{F}$ are completely positive trace preserving maps (CPTP) acting on density operators $\mathcal{D}$ and composing in the usual sense of composition of maps.

	Free operations naturally induce a partial order on the state space:
	
	\begin{definition}[Partial order induced by $\mathcal{F}^*$]
	Given a set of free operations $\mathcal{F}^*$ and two states $\rho,\sigma \in \mathcal{D}$, we say that $\sigma \leqF \rho$, if there exists a $\phi \in \mathcal{F}^*$ such that $\phi(\rho)=\sigma$.
	\end{definition}
	
    Furthermore, if two states can be transformed into  one another through free operations, the two states are  indistinguishable to the resource theory. This motivates the following definition. 
	
	\begin{definition}[Quotient space]
		Any resource theory naturally induces an equivalence relation on the space of states as:
		\begin{align}
			\label{eq:equivalence}
			\rho\simeq_{\mathcal{F}} \sigma \iff \sigma \leqF \rho \mbox{ } \land \mbox{ } \rho \leqF \sigma.
		\end{align}
		The natural space in which the resource theory is defined is  $\mathcal{D}/\simeq_{\mathcal{F}} $.
	\end{definition}
	
	In order to quantify the value of a state, one introduces functions which cannot increase under free operations. 
	
	\begin{definition}[Monotones and compatibility]
		A function is called monotone (w.r.t. a resource theory) when it does not increase under the action of the free operations. Respectively, an operation $\phi$ is said to be compatible with a set of monotones $\mathcal{M}$ if $\forall f \in \mathcal{M}, \forall \rho \in \mathcal{D} \mbox{ } f(\phi(\rho))\leq f(\rho)$. 
	\end{definition}

	\begin{definition}[Complete set of monotones]		
		A set of monotones $\mathcal{M}_{\mathcal{F}}$ is said to be complete (w.r.t $\mathcal{F}^*$) if  
		\begin{align}
			\label{def:complete-monotones} 
			\forall f \in \MF, \forall \rho,\sigma \in \mathcal{D}, f\left(\sigma\right) \leq f(\rho)  \iff \sigma \leqF \rho
		\end{align}
		That is, $\MF$ and $\mathcal{F}^*$ identifies the same partial order structure on $\mathcal{D}$.
	\end{definition}

	\section{Main results}
	
	In this section we prove that deciding whether a transition is present in a resource theory is undecidable (Corollary~\ref{lemma:membershipMonoids}), which implies that one cannot tell whether two resource theories presented in terms of free operations are the same (Corollary~\ref{lemma:distinguishability}). Moreover, we also show that it is impossible to algorithmically construct from a set of free operations $\mathcal{F}^*$ a set of monotones $\mathcal{M}$ which describes the same resource theory (Corollary~\ref{thm:resourcevsfreeoperations}). These facts are all consequences of the following:
	\begin{theorem}\label{thm:undecidability}
		The membership problem for semigroups of CPTP maps is undecidable.
	\end{theorem}
	
	\begin{proof}
		The main arguments of the proof are inspired by~\cite{bellMembershipInvertibleDiagonal2007, bellReachabilityProblemsQuaternion2008}. Define two generic matrices in SU$(2)$:
		\begin{align}
			A = e^{i \theta \vec{n}\cdot\vec{\sigma}}, \qquad B = e^{i \theta \vec{m}\cdot\vec{\sigma}}
		\end{align}
		where $\vec{n}$ and $\vec{m}$ are vectors in $\mathbb{R}^3$ and $\vec{\sigma}:= \{\sigma_x,\sigma_y,\sigma_z\}$ is a vector of Pauli matrices. It was proven in~\cite{swierczkowskiClassFreeRotation1994} that any pair $\{A,B\}$ of this form generates a free semigroup whenever  $\vec{n}\cdot \vec{m} = 0$  and $\cos\theta \in {\mathbb{Q}\backslash\{0,\pm1,\pm\frac{1}{2}\}}$. A semigroup of two elements is called free if there is a bijection between its elements and binary strings or, equivalently, if there is no finite composition of its elements that gives the identity. We can then use the two matrices above to encode words in $\{0,1\}^*$, where the star indicates arbitrary finite juxtaposition of the letters in a given set. Define $\gamma$ as the homomorphism that assigns to each binary word the corresponding element in the semigroup $\{A,B\}^*$. In other words, $\gamma$ operates on binary strings by substituting to each $0$ an $A$ and to each $1$ a $B$, and the juxtaposition of letters is mapped to matrix multiplication (e.g., $\gamma(010) = ABA$). Finally, we also need to define a matrix $C$ not in $\{A,B\}^*$, which can be easily done by choosing a unitary that squares to the identity. Since $\{A,B\}^*$ is free, by construction $C\notin \{A,B\}^*$.
		
		\begin{figure}
			\centering
			\includegraphics[width=0.9\linewidth]{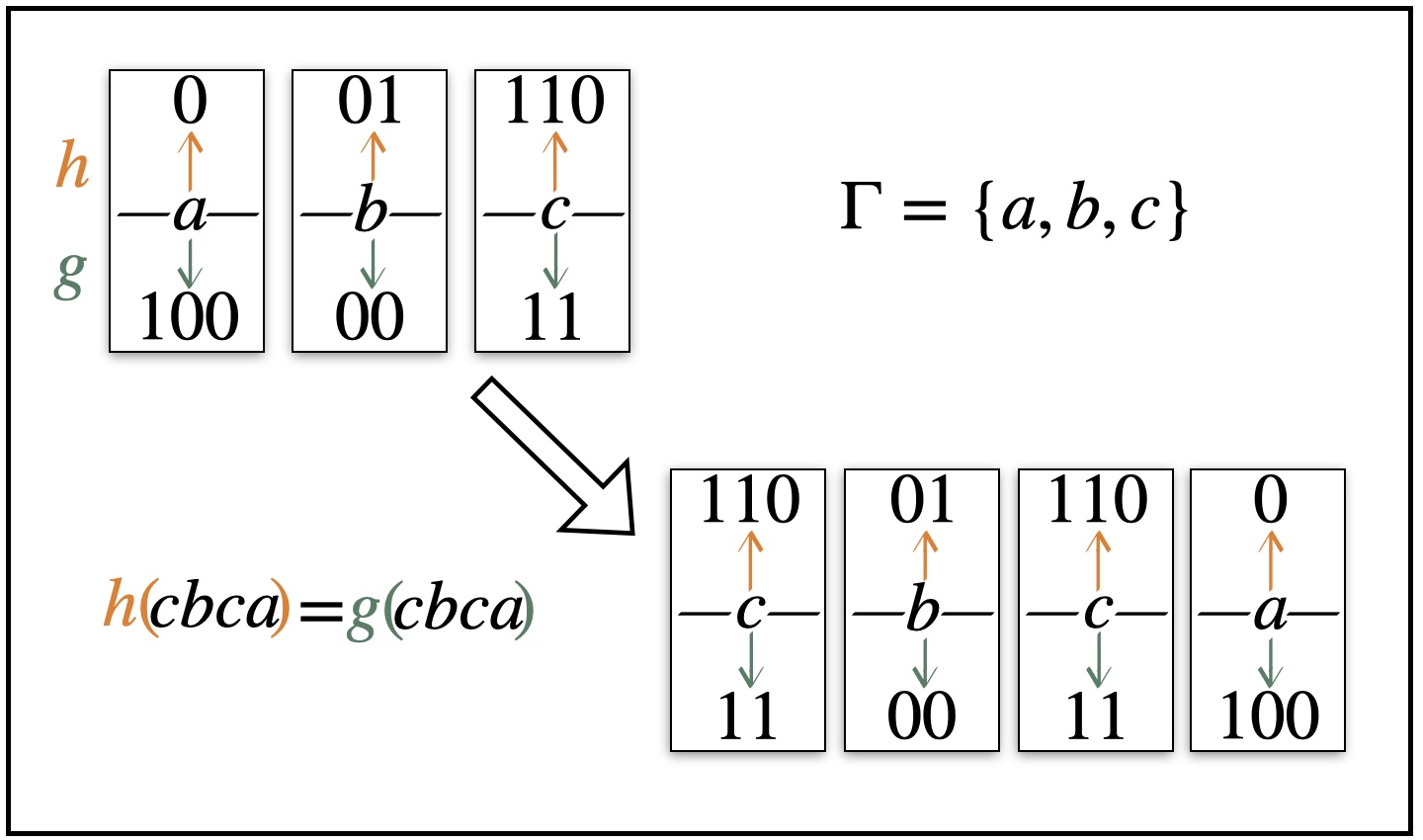}
			\caption{The PCP can be expressed in terms of dominoes: given the two homomorphisms $h$ and $g$ from $\Gamma^*$ to $\{0,1\}^*$, we assign to each letter $x\in\Gamma$ a tile which has in the upper half $h(x)$ written out, and in the lower half $g(x)$. In the example above, $h(a) = 0$ and $g(a)=100$, so that the $a$-domino is constructed accordingly. Since $h$ is an homomorphism (i.e., $h(xy)=h(x)h(y)$), words in $\Gamma^*$ can be represented by juxtaposing the dominoes corresponding to each letter in the word. The PCP problem then translates to the question whether one can find a sequence of dominoes that makes the binary word appearing on top equal to the one on the bottom.}
			\label{fig:figpcp}
		\end{figure}

		In order to prove the theorem we use a reduction to the Post Correspondence Problem (PCP). Given two different homomorphisms $h$ and $g$ from the finite alphabet $\Gamma$ to $\{0,1\}^*$, this is the problem of deciding whether there is a non-empty word $w\in\Gamma^*$ such that $h(w) = g(w)$. It is a classical result from computability theory that the PCP is undecidable~\cite{postVariantRecursivelyUnsolvable1946, sipser2012introduction}. Moreover, if one restricts the problem to words whose first character is fixed, this is also undecidable. The idea of the proof is to show that if the statement of the theorem were decidable, then the PCP would be as well, generating a contradiction. 
		
		First, for each letter $a_i\in\Gamma$ define the two unitary matrices:
		\begin{align}
			&h_{a_i} =  \begin{pmatrix}
				\gamma(h(a_i)) & 0_2\\
				0_2 & A^{i} B
			\end{pmatrix}\,,\;\; g_{a_i} =  \begin{pmatrix}
			\gamma(g(a_i))^\dagger & 0_2\\
			0_2 &   (A^{i})^\dagger B^\dagger
		\end{pmatrix}
		\end{align}
		Since both $\gamma$ and $h$ are homomorphisms, the matrices $h_{a_i}$ compose as $h_{a_i}h_{a_j} = h_{a_i a_j}$ (similarly for $g_{a_i}$). These matrices are constructed in such a way  that the first diagonal block encodes the image of the letter $a_i$ under the homomorphism, while the bottom half is used to keep track of its index. For technical reasons that will become clear in a moment, we also add an extra matrix $\tilde{g}_{a_1}$ giving an alternative encoding of $a_1$, having $\gamma(g(a_1))^\dagger C$ in the upper left corner, and everything else equal to $g_{a_1}$. Moreover, we also introduce two extra unitaries, which will keep track of the beginning and the end of the string, given by:
		\begin{align}
			s =  \begin{pmatrix}
				C & 0_2\\
				0_2 &  B
			\end{pmatrix}\,,\;\; f =  \begin{pmatrix}
				\id_2 & 0_2\\
				0_2 &   B^\dagger
			\end{pmatrix}\,.
		\end{align}
		
		We can then study the resource theory generated by the set $\mathcal{F} := \{\id, H^\lambda_{a_i}, G^\lambda_{a_i},S, F\}_{a_i\in\Gamma, \lambda\in(0,1)}$, where we defined the following maps:
		\begin{align}			\label{eq:EncodingPCP}
			\begin{cases}
				H^\lambda_{a_i}(\rho):=\lambda\,h_{a_i}\rho\,h_{a_i}^\dagger + (1-\lambda) \frac{\id}{4}\\
				G^\lambda_{a_i}(\rho):=\lambda\,g_{a_i}\rho\,g_{a_i}^\dagger + (1-\lambda) \frac{\id}{4}\\
				S(\rho):=\tilde{\lambda}\, s\,\rho\,s^\dagger +(1-\tilde{\lambda}) \frac{\id}{4}\\
				F(\rho):=\tilde{\lambda}\,f\,\rho\,f^\dagger + (1-\tilde{\lambda})\frac{\id}{4}\,,
			\end{cases}
		\end{align}  
		for some fixed $\tilde{\lambda}$.
		The composition in this case behaves as $H^{\lambda_1}_{a_i}H^{\lambda_2}_{a_j} = H^{\lambda_1\lambda_2}_{a_i a_j}$, and similarly with the other elements. In this way, one can encode words from $\Gamma^*$  into operations in $\mathcal{F}^*$ constructed by composing either only $H^\lambda_{a_i}$s or $G^\lambda_{a_i}$s.
		
		We are now ready to prove the claim. Consider the operation:
		\begin{align}\label{eq:rhodef}
			\psi(\rho):=\lambda \,\rho  + (1-\lambda) \frac{\id}{4}\,
		\end{align}
		with $\lambda\in (\tilde{\lambda}^3,\tilde{\lambda}^2)$. Deciding whether $\psi\in\mathcal{F}^*$ is equivalent to the PCP. First, notice that the constraint on $\lambda$ forces the total number of $S$ and $F$ to be at most two. In this context, the only way to obtain the identity in the second diagonal block is having compositions of the form $g_{w_n}\dots g_{w_2}g_{w_1} s \,h_{w_1}h_{w_2}\dots h_{w_n} f$ (or any cyclic composition thereof). Moreover, $w_1$ has to coincide with $a_1$, because $\tilde{g}_{a_1}$ is the only matrix containing an instance of $C$, and $C\notin\{A,B\}^*$. Thus, given that the $h$ matrices and the $g$ matrices cluster in two different groups, in order to get the identity in the first diagonal block the following should hold
		\begin{align}
			\gamma(g(w_n))^\dagger\dots \gamma(g(w_1))^\dagger \,\gamma(h(w_1))\dots \gamma(h(w_n)) = \id_2,
		\end{align} 
		where $w_1 = a_1$ (and we use the encoding $\tilde{g}_{a_1}$ for the first character only),
		which is equivalent to the existence of a word $w\in\Gamma^*$ starting with $a_1$ such that $h(w)  = g(w)$. This reduces the PCP to the membership problem for CPTP maps.
	\end{proof}
	
	The theorem just proved has a number of implications. In particular:
	
	\begin{corollary}[Reachability problem]\label{lemma:membershipMonoids}
		Given two states $\rho$ and $\sigma$ and the description of $\mathcal{F}$, it is undecidable whether there exists $\phi\in\mathcal{F}^*$ such that $\phi(\rho) = \sigma$.
	\end{corollary}
	\begin{proof}
		This follows directly from the proof of Theorem~\ref{thm:undecidability}: choose an arbitrary $\rho$ and set $\sigma$  to be $\psi(\rho)$ (i.e, the state defined by the right hand side of Eq.~\eqref{eq:rhodef}). Due to the structure of the semigroup $\mathcal{F}^*$, a transition between the two states is possible if and only if ${\psi\in\mathcal{F}^*}$. Since this is undecidable, the corollary follows.
	\end{proof}
	 
	Moreover, it is also easy to see that:
	
	\begin{corollary}\label{lemma:distinguishability}
		Given two generating sets $\mathcal{F}_1$ and $\mathcal{F}_2$, there is no algorithmic means of deciding whether they describe the same resource theory.
	\end{corollary}
	\begin{proof}
		Take as $\mathcal{F}_1$ the set $\mathcal{F}$ defined in Theorem~\ref{thm:undecidability} and as $\mathcal{F}_2 := \mathcal{F}\cup \{\psi\}$. The transition $\rho\rightarrow\psi(\rho)$ is trivially present in $\mathcal{F}_2$. From Corollary~\ref{lemma:membershipMonoids}, though, it is undecidable to say whether  this transition is generated by $\mathcal{F}_1$. Hence, deciding whether two different sets of free operations describe the same set of transitions is impossible in general.
	\end{proof}
	
	\begin{comment}
	\begin{corollary}
	The quotient space $\mathcal{D}/\simeq_{\mathcal{F}} $ is not recursively enumerable in general.
	\end{corollary}
	\begin{proof}
	Assume there was a way to decide whether the equivalence class of $\rho$ was trivial or not. This is in one-to-correspondence with the existence in $\mathcal{F}$ of a map $\phi\neq\id$ such that $\phi(\rho) = \rho$. But this is undecidable.
	\end{proof}
	\end{comment}
	\begin{figure}
		\centering
		\includegraphics[width=0.9\linewidth]{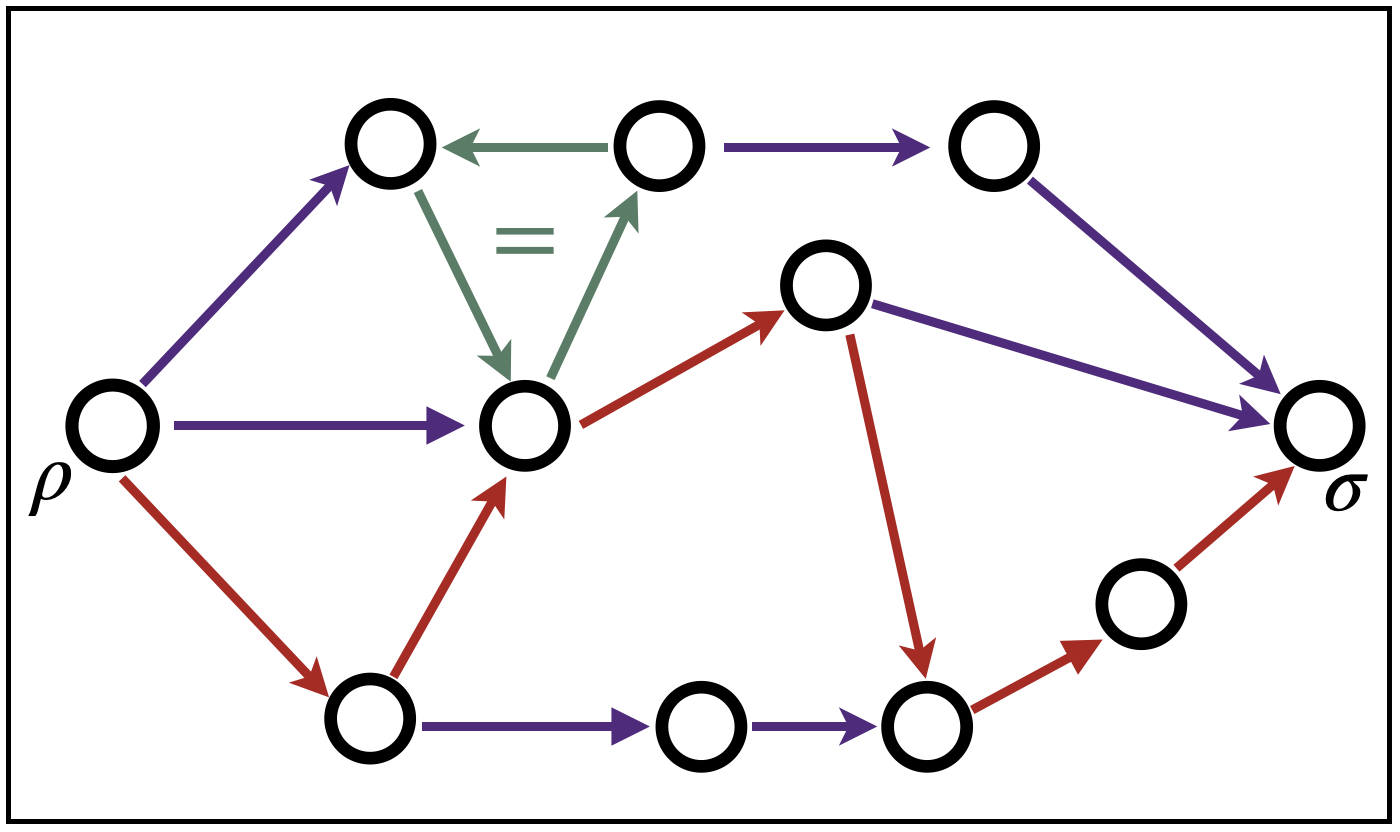}
		\caption{Graphical representation of the action of a resource theory $\mathcal{F}^*$ on the state $\rho$. Each vertex represents a state, and the presence of an edge from the vertex $\rho_1$ to $\rho_2$ corresponds to the existence of a transformation $\phi\in\mathcal{F}^*$ such that $\phi(\rho_1) = \rho_2$. It should be noticed that in order to keep the graph clean, we omitted many edges in the figure. For example, there should be a line directly connecting $\rho$ to $\sigma$, since $\mathcal{F}^*$  contains all the compositions. The green arrows show that equivalent states appear in the graph in the form of a cycle. Collapsing all the cycles to a representative vertex makes the graph acyclic. The red line is the longest trajectory between $\rho$ and $\sigma$, so that $f_{\rho}^\mathcal{F}(\sigma) = \frac{1}{7}$.
			% In panel 2 we show how the graph can be iteratively reconstructed. At each step, one adds all the vertices reachable within at most a single transformation. The numbers in the figure show at which step of the algorithm each arrow is added to the graph. In this way, all the trajectories can be inductively retrieved.
		}
		\label{fig:plot1}
	\end{figure}
	
	Before passing to prove Corollary~\ref{thm:resourcevsfreeoperations}, we present the following:
	
	\begin{lemma}\label{lemma:completeMonotones}
		For any $\mathcal{F}^*$ there exists a complete set of monotones.
	\end{lemma}
	\begin{proof}
		It is natural to first define a set of monotones  $\mathcal{M}$ on the quotient space $\mathcal{D}/\simeq_{\mathcal{F}}$ and to just later extend it to the whole space of density matrices. 
		
		The idea of the proof is to assign to each state $\rho\in\mathcal{D}$ a directed graph corresponding to all the possible states reachable from $\rho$ through arbitrary applications of elements of $\mathcal{F}^*$. A graphical depiction of how this looks like is presented in the first panel of Fig.~\ref{fig:plot1}. Each vertex in the figure corresponds to a state, and the presence of an arrow from the vertex $\rho_1$ to $\rho_2$ corresponds to the existence of a transformation $\phi\in\mathcal{F}^*$ such that $\phi(\rho_1) = \rho_2$. Passing from $\mathcal{D}$ to the quotient space $\mathcal{D}/\simeq_{\mathcal{F}}$ makes the graph acyclic, since all the vertices in the same equivalence class (like the ones connected by the green lines in the figure) collapse to a single point. 
		
		In order to define monotones on $\mathcal{D}/\simeq_{\mathcal{F}}$, we then only have to focus on acyclic directed graphs. Moreover, one has to assign to each continuous semigroups  ${\{\Phi_t\}_{t\in\R}}\subset\mathcal{F}^*$ a generator $\mathcal{L}$, i.e., a map satisfying $\Phi_t =: e^{t\mathcal{L}}$. At this point, a length can be assigned to each edge in the graph: if the transformation corresponding to an edge cannot be further decomposed inside $\mathcal{F}^*$, then we assign a length one; if the transformation is part of a semigroup generated by a single element $\mathcal{L}$, the length is given by the parameter $t$; finally, if an edge corresponds to multiple transformations, the length is inductively defined as the sum of the lengths. We can then assign to each  state $\rho\in\mathcal{D}_C$ a monotone $f_{\rho}^\mathcal{F}$ in the following way: (i) if a state $\sigma$ cannot be reached from $\rho$, then $f_{\rho}^\mathcal{F}(\sigma):=2$; (ii) if $\sigma$ is part of at least one of the trajectories stemming from $\rho$,  then assign the value $f_{\rho}^\mathcal{F}(\sigma):=\frac{1}{\ell+1}$, where $\ell$ is the length of the longest path from $\rho$ to $\sigma$. Finally, $f_{\rho}^\mathcal{F}$ can be extended to the whole space of density matrices, by assigning to two representative of the same equivalence class $\sigma_1\simeq_{\mathcal{F}}\sigma_2$ the same value, $f_{\rho}^\mathcal{F}(\sigma_1) \equiv f_{\rho}^\mathcal{F}(\sigma_2)$. 
		
		In order to prove that the set $\mathcal{M} = \{f_{\rho}^\mathcal{F}\}_{\rho\in\mathcal{D}}$ is complete, we first need to prove that it is compatible. This holds by construction. 	In fact, suppose that there exists a $f_{\rho}^\mathcal{F}$, a $\sigma$ and a $\phi\in\mathcal{F}^*$ such that $f_{\rho}^\mathcal{F}(\phi(\sigma))> f_{\rho}^\mathcal{F}(\sigma) := \frac{1}{s+1}$.  The longest trajectory from $\rho$ to $\phi(\sigma)$ is long at least $s+\varepsilon$ for some positive $\varepsilon$, implying that $f_{\rho}^\mathcal{F}(\phi(\sigma))\leq \frac{1}{s+\varepsilon+1}<\frac{1}{s+1}$, which gives the desired contradiction. Similarly, it also follows  by construction that the set is complete. Assume that for two given states $\rho$ and $\sigma $ and $\forall f \in \mathcal{M}$ one has $f(\sigma)\leq f(\rho)$. This directly implies that $f_\rho^\mathcal{F}(\sigma)\leq f_\rho^\mathcal{F}(\rho) = 1$, so by definition there exists a trajectory from $\rho$ to $\sigma$ or, in other words, a $\phi\in\mathcal{F}^*$ such that $\phi(\rho) = \sigma$. This concludes the proof.
	\end{proof}
	
	Lemma~\ref{lemma:completeMonotones} shows that monotones can be as powerful in constraining a resource theory as the usual characterisation in terms of free operations. Since we are just interested about the possibility in principle, we constructed an overcomplete set $\mathcal{M}$, setting aside questions about finding the minimal complete set. It is worth pointing out, though, that in many resource theories the complete set is actually finite: for example, in the resource theory of non-uniformity one only needs $d-1$ monotones, where $d$ is the dimensionality of the Hilbert space~\cite{gourResourceTheoryInformational2015}. Before investigating how one could find a minimal complete set, it is important to understand whether such a  complete set could be found at all. The negative answer to this issue is given by the following:
	
	\begin{corollary}\label{thm:resourcevsfreeoperations}
		Given a set of free operations $\mathcal{F}$ there is no algorithmic means of constructing a recursive complete set of monotones $\mathcal{M}$ associated to it.
	\end{corollary}
	\begin{proof}
		Lemma~\ref{lemma:completeMonotones} implies the existence of a complete set of monotones $\mathcal{M}$ associated to $\mathcal{F}$. Assume the existence of an algorithm constructing this set. Moreover, also assume  the existence of an algorithm $\mathcal{A}_{\rm cmp}$ which takes two arbitrary states $\rho$ and $\sigma$ as input, and decides whether $\forall f\in\mathcal{M}$, $f(\rho)\geq f(\sigma)$. If such an algorithm did not exist, then the claim would follow trivially, since one wouldn't be able to check the compatibility of $\mathcal{F}$ with $\mathcal{M}$. If such an algorithm existed, then one could decide whether a transition between two arbitrary states is present in $\mathcal{F}^*$ simply by running $\mathcal{A}_{\rm cmp}$. This is in contradiction with Corollary~\ref{lemma:membershipMonoids}. Therefore, there is no general algorithm constructing the complete set $\mathcal{M}$ associated with $\mathcal{F}$.
	\end{proof}
	
	\section{Conclusions}
	
	In the present letter we showed that standard questions in resource theories hide undecidability issues. Consider for example the main problem in this context: whether it is possible to convert a state into another through a series of allowed operations. The undecidability of this question is the content of Corollary~\ref{lemma:membershipMonoids}. This directly implies the impossibility of completely identifying a generic resource theory. In particular, given two different characterisations of a resource theory in terms of two different sets of free operations, it is impossible to tell whether they induce the same set of transitions (Corollary~\ref{lemma:distinguishability}). Finally,  Corollary~\ref{thm:resourcevsfreeoperations} implies that it is also impossible to construct a set of value functions which completely describes the resource theory induced by a set of free operations.
	
	The main theorem supporting these results is the undecidability of the membership problem for semigroups of CPTP maps (Theorem~\ref{thm:undecidability}). It should be noticed that in the proof of this theorem we used a resource theory which, if somewhat artificial, has anyway a  precise physical interpretation: the maps $H_{a_i}^\lambda$ and $G_{a_i}^\lambda$ defined in Eq.~\eqref{eq:EncodingPCP} are rotated depolarising channels, which correspond to the experimental setting in which one can only apply  unitary transformations of the form $h_{a_i}$ or $g_{a_i}$, introducing at each application some quantum noise into the system. 
	
	 Nonetheless, it would be interesting to prove the same theorem within the framework of a more natural resource theory. A promising candidate in this respect is given by  the LOCC set: its structure is notoriously difficult to characterise mathematically~\cite{chitambarEverythingYouAlways2014}, fact that could be explained if one could show that it contains an undecidable set. This challenging possibility would require the devise of a new proof and it is left for future research.
	 
	%%%%%%%%%%%%%%%%%%%%%%%%%%%%%%%%%%%%%%%%%%
	
	\emph{Acknowledgements.} We especially thank G. Senno and R. Ganardi  for pointing out mistakes in the preliminary and the published version, respectively. We also thank M. Lostaglio and A.~Acín for the useful comments, and A. Hermes-Müller for pointing out the reference~\cite{wolfAreProblemsQuantum2011}, where a slightly stronger version of Corollary~\ref{lemma:membershipMonoids} was proven through  a different reduction of the PCP (Theorem 2 therein). 
	This project has received funding from the European Union’s Horizon 2020 research and innovation programme under the Marie Skłodowska-Curie grant agreement No 713729,
	and from the Government of Spain (FIS2020-TRANQI and Severo Ochoa CEX2019-000910-S), Fundacio Cellex, Fundació Mir-Puig, Generalitat de Catalunya (SGR 1381 and CERCA Programme).
	
	\bibliography{bib.bib}
\end{document}